\documentclass[11pt,twocolumn]{article}

\usepackage[T1]{fontenc}
\usepackage[utf8]{inputenc}
\usepackage{lmodern}
\usepackage[margin=0.70in]{geometry}
\usepackage{microtype}
\usepackage{amsmath,amssymb,amsthm,mathtools}
\usepackage{aliascnt}
\usepackage{algorithm}
\usepackage[noend]{algpseudocode}
\usepackage{booktabs}
\usepackage{tikz}
\usetikzlibrary{arrows.meta,calc,positioning}
\usepackage[numbers,sort&compress]{natbib}
\usepackage{titling}
\usepackage{titlesec}
\usepackage[font=small,skip=4pt]{caption}
\usepackage{hyperref}
\usepackage{doi}
\usepackage{orcidlink}
\usepackage[nameinlink,noabbrev]{cleveref}
\hypersetup{
  colorlinks=true,
  linkcolor=blue,
  filecolor=magenta,
  urlcolor=cyan,
  citecolor=orange,
  pdftitle={Random-Priority Frontier Routing: Tight Theta(n\string^c) Bounds Against c-Node Cartels},
  pdfauthor={Krisjanis Petrucena}
}

\posttitle{\par\end{center}\vspace{-0.4em}}
\postauthor{\end{tabular}\par\end{center}\vspace{-1.5em}}
\predate{}
\postdate{}

\titlespacing*{\section}{0pt}{1.6ex plus .2ex minus .2ex}{0.9ex plus .2ex}
\titlespacing*{\subsection}{0pt}{1.3ex plus .2ex minus .2ex}{0.6ex plus .2ex}

\makeatletter
\def\thm@space@setup{%
  \thm@preskip=0.6\baselineskip
  \thm@postskip=\thm@preskip}
\makeatother

\AtBeginDocument{%
  \setlength{\abovedisplayskip}{5pt plus 2pt minus 2pt}%
  \setlength{\belowdisplayskip}{5pt plus 2pt minus 2pt}%
  \setlength{\abovedisplayshortskip}{2pt plus 2pt}%
  \setlength{\belowdisplayshortskip}{3pt plus 2pt minus 2pt}%
}

\newtheorem{theorem}{Theorem}
\newaliascnt{proposition}{theorem}
\newtheorem{proposition}[proposition]{Proposition}
\aliascntresetthe{proposition}
\newaliascnt{corollary}{theorem}
\newtheorem{corollary}[corollary]{Corollary}
\aliascntresetthe{corollary}
\newaliascnt{lemma}{theorem}

\aliascntresetthe{lemma}
\theoremstyle{definition}
\newaliascnt{definition}{theorem}

\aliascntresetthe{definition}
\theoremstyle{remark}
\newaliascnt{remark}{theorem}

\aliascntresetthe{remark}

\newcommand{\RPFR}{\textnormal{\textsc{rpfr}}}
\newcommand{\HS}{\textnormal{\textsc{hs}}}

\DeclareMathOperator*{\argmax}{arg\,max}

\title{\textbf{Random-Priority Frontier Routing:}\\
Tight \(\boldsymbol{\Theta(n^c)}\) Bounds Against \(\boldsymbol{c}\)-Node Cartels}
\author{Kri\v{s}j\=anis Petru\v{c}e\c{n}a\,\orcidlink{0009-0008-5713-5914}\\
\small Institute of Mathematics and Computer Science of the University of Latvia\\
\small \texttt{krisjanis.petrucena@lumii.lv}}
\date{}

\begin{document}
\raggedbottom
\maketitle

\begin{abstract}
We study path diversification in trusted-node networks, where sensitive material is relayed through intermediate nodes, some of which may be compromised.
Our randomized routing rule assigns each vertex an independent random priority and repeatedly expands the highest-priority vertex on the global frontier of the explored region.
Let \(G\) have \(n\) vertices, let \(s,t\) be honest endpoints, and let \(C\) be a set of \(c\) compromised intermediate vertices, called a cartel, whose deletion leaves \(s\) and \(t\) connected.
For every fixed \(c\) and every fixed target probability \(q\in(0,1)\), we prove that \(\Theta(n^c)\) independent executions are sufficient in the worst case for some route to avoid \(C\) with probability at least \(q\).
\end{abstract}

\section{Introduction}
\label{sec:introduction}

Trusted-node quantum key distribution (QKD) networks relay key material when the communicating endpoints do not share a direct quantum link.
Hop-by-hop encryption does not prevent a compromised intermediate node from observing relayed material, so path diversification and privacy amplification are used to reduce reliance on any one relay \citep{salvail2010trusted,dervisevic2025keymanagement}.

Our preceding work studied topology-oblivious random-walk key relaying and introduced a highest-score-neighbor rule, denoted \(\HS\) \citep[Sec.~3.8]{petrucena2026topology}.
For each key material fragment, a fresh seed defines scores of vertex identifiers, and a relay node forwards the fragment to its highest-scored unvisited neighbor.
The rule was evaluated empirically rather than given a graph-independent cartel-avoidance guarantee.
\Cref{prop:local-failure} shows that its \emph{avoidance} probability can be exponentially small.

This note asks what changes when the same random ranking drives \emph{global-frontier} exploration.
Starting from \(s\), the process maintains a connected visited set \(S\); at every step it selects the highest-ranked unvisited vertex adjacent to \(S\), moves through the visited subgraph to reach an adjacent visited vertex, and expands \(S\).
We call the resulting procedure \emph{random-priority frontier routing} (\(\RPFR\)).
Global extremal-frontier growth is familiar from bond and site invasion percolation \citep{wilkinson1983invasion,chayes1985stochastic,glantz2008invasion}.
\(\RPFR\) applies that mechanism to an abstract vertex ranking with an adversarial objective and a designated stopping target.
Greedy random walk \citep{orenshtein2014greedy} is a different local rule, selecting uniformly among untraversed incident edges, and its cover-time results imply no cartel-avoidance bound.

Let \(G\) have \(n\) vertices, let \(s,t\) be honest, and let \(C\) be a static set of \(c\) intermediate vertices whose deletion does not disconnect \(s\) from \(t\).
Then
\[
  \Pr(\text{\(\RPFR\) avoids \(C\)})
  \geq \binom{n-1}{c}^{-1},
\]
and we construct a graph attaining equality.
If executions use fresh independent rankings and \(C\) is fixed, the worst-case number required for a constant probability of at least one cartel-free route is therefore \(\Theta\!\binom{n-1}{c}\), which is \(\Theta(n^c)\) for every fixed \(c\).
This is a routing guarantee: it bounds whether a route visits a cartel vertex, not whether an avoided route yields secret key.

\section{Model}
\label{sec:model}

Let \(G=(V,E)\) be a simple undirected graph with \(n=\lvert V\rvert\), and let \(s,t\in V\) be the distinct source and target.
A \emph{cartel} is a set \(C\subseteq V\setminus\{s,t\}\) of \(c\) vertices, where \(0\leq c\leq n-2\).
We call \(C\) \emph{admissible} if \(G-C\) contains an \(s\)--\(t\) path.
Cartels that separate \(s\) from \(t\) are excluded because no routing procedure can produce an \(s\)--\(t\) route avoiding such a cut.

The cartel is \emph{static}: it is selected independently of the execution's random priorities and is not changed after priorities are revealed.
Relayed material is intercepted exactly when the route visits a vertex of \(C\) before first visiting \(t\).

The routing procedure does not know \(G\) in advance.
Adjacency is revealed only from vertices that have already been visited, as in the topology-oblivious setting of the predecessor protocol.
Authenticated vertex identifiers are fixed before the fresh execution randomness is chosen, and scores are generated by a trusted controller or derived by it from the fresh seed and those identifiers, so a cartel cannot choose its scores after learning honest scores.
Discovered adjacency and frontier data are trusted.
Cartel vertices may observe material that visits them, but cannot forge priorities, impersonate vertices, suppress honest edges from the discovered frontier, or alter the algorithm.
Without these control-plane assumptions, a malicious relay could bias frontier selection and the ranking argument below would not apply.

\subsection{Random-priority frontier routing}

Before an execution, assign each vertex \(u\in V\setminus\{s\}\) a random score \(R(u)\), independent and identically distributed according to a continuous distribution.
Only relative score order matters.
The i.i.d.\ continuous scores induce a uniformly random ordering of \(V\setminus\{s\}\), so \(R(s)\) need not be sampled.
For \(S\subseteq V\), let \(G[S]\) denote the subgraph induced by \(S\), and let \(F(S)=N(S)\setminus S\) denote its open vertex frontier.
The procedure is given in \cref{alg:rpfr}.

\begin{algorithm}[H]
  \caption{Random-priority frontier routing}
  \label{alg:rpfr}
  \begin{algorithmic}[1]
    \Require Graph access revealed from visited vertices, connected endpoints \(s,t\), scores \(R\)
    \State \(S\gets\{s\}\)
    \While{\(t\notin S\)}
      \State \(v\gets\argmax_{u\in F(S)}R(u)\)
      \State choose any \(w\in S\cap N(v)\)
      \State move from the current location to \(w\) along a path in \(G[S]\)
      \State traverse \(\{w,v\}\) and set \(S\gets S\cup\{v\}\)
    \EndWhile
    \State \Return the exploration trace
  \end{algorithmic}
\end{algorithm}

Each selected frontier vertex has a neighbor in \(S\), so \(G[S]\) remains connected and relocation within \(S\) is always possible.

If \(s\) and \(t\) are connected, the frontier is nonempty whenever \(t\notin S\), and each iteration adds one vertex, so \(\RPFR\) visits \(t\) after at most \(n-1\) expansions.
The choice of \(w\) and of the path inside \(G[S]\) changes traversal cost but not the order in which new vertices enter \(S\).
The target receives no direct-neighbor preference: it is selected only when its score is maximal on the whole frontier.
Consequently \(R(t)\) participates in the random ordering, which is why the theorem below contains \(\binom{n-1}{c}\) rather than \(\binom{n-2}{c}\).

An implementation can use a fresh per-execution seed and derive each score from the seed and authenticated vertex identifier when the vertex first enters the frontier, without knowing \(V\) in advance, adapting the per-fragment \(\HS\) construction of \citet[Sec.~3.8]{petrucena2026topology}.\footnote{That replacement is computational and should not be called information-theoretic randomness.}

\section{Tight cartel-avoidance bound}
\label{sec:tight-bound}

For a fixed instance \((G,s,t,C)\), let \(p_{\mathrm{avoid}}(G,s,t,C)\) denote the probability that \(t\) is visited before any vertex of \(C\).

\begin{theorem}[Universal avoidance bound]
\label{thm:lower-bound}
For every \(n\)-vertex graph \(G\), every distinct pair \(s,t\in V\), and every admissible \(c\)-node cartel \(C\),
\[
  p_{\mathrm{avoid}}(G,s,t,C)
  \geq
  \binom{n-1}{c}^{-1}.
\]
\end{theorem}

\begin{proof}
If \(c=0\), then \(C=\varnothing\), both sides equal one, and the result is immediate.
Assume \(c\geq 1\), let \(H=V\setminus(C\cup\{s\})\) be the \(n-1-c\) honest non-source vertices, including \(t\), and consider the event
\[
  \mathcal{E}
  =
  \Bigl\{
    \max_{v\in C}R(v)<\min_{u\in H}R(u)
  \Bigr\}.
\]
Among the \(n-1\) ranked non-source vertices, \(\mathcal{E}\) occurs exactly when the fixed set \(C\) occupies the bottom \(c\) positions.
Every \(c\)-subset is equally likely to occupy those positions, hence \(\Pr(\mathcal{E})=\binom{n-1}{c}^{-1}\).

It remains to show that \(\mathcal{E}\) implies avoidance.
Suppose that the current visited set \(S\) contains \(s\), contains no cartel vertex, and does not contain \(t\).
Because \(C\) is admissible, fix an \(s\)--\(t\) path in \(G-C\).
The first vertex of this path outside \(S\) has a predecessor in \(S\), so it is an honest vertex in \(F(S)\).
On \(\mathcal{E}\), every honest frontier vertex has a higher score than every cartel vertex, so the maximizer selected by \(\RPFR\) is honest.
Starting with \(S=\{s\}\), induction shows that no cartel vertex is visited before \(t\).
Thus \(\mathcal{E}\) is contained in the avoidance event.
\end{proof}

The sufficient event in the proof may be stronger than necessary on a particular graph.
The next construction shows that no larger graph-independent guarantee is possible.

\begin{theorem}[Tightness]
\label{thm:tightness}
For every \(n\geq 2\) and \(0\leq c\leq n-2\),
the minimum below ranging over all simple undirected \(n\)-vertex graphs, distinct endpoints, and admissible \(c\)-node cartels satisfies
\[
  \min p_{\mathrm{avoid}}(G,s,t,C)
  =
  \binom{n-1}{c}^{-1}.
\]
\end{theorem}

\begin{proof}
The lower bound is \cref{thm:lower-bound}, and for \(c=0\) every execution avoids the empty cartel, so equality holds.
Assume \(c\geq 1\), let \(q=n-c-1\), and take the honest path \(s=v_0,v_1,\ldots,v_q=t\).
For each \(v\in C\), add the edges \(\{s,v\}\) and \(\{v,v_1\}\), as in \cref{fig:tight-instance}.
Deleting \(C\) leaves the honest path, so \(C\) is admissible.

Until either \(t\) or a cartel vertex is visited, the frontier consists of all vertices of \(C\) and the next unvisited path vertex \(v_i\).
Therefore \(\RPFR\) avoids \(C\) if and only if \(R(v_i)>\max_{v\in C}R(v)\) for every \(i\in\{1,\ldots,q\}\).
Equivalently, all vertices of \(C\) occupy the bottom \(c\) positions in the random ordering of \(V\setminus\{s\}\), an event of probability \(\binom{n-1}{c}^{-1}\).
\end{proof}

\begin{figure}[H]
  \centering
  \begin{tikzpicture}[
    scale=0.85,
    vertex/.style={circle,draw,minimum size=6mm,inner sep=0pt,font=\small},
    cartel/.style={circle,draw,dashed,minimum size=6mm,inner sep=0pt,font=\small},
    >=Latex
  ]
    \node[vertex] (s) {\(s\)};
    \node[vertex,right=15mm of s] (v1) {\(v_1\)};
    \node[right=12mm of v1] (dots) {\(\cdots\)};
    \node[vertex,right=12mm of dots] (vq) {\(v_q=t\)};
    \node[cartel,above=13mm of $(s)!0.5!(v1)$] (b1) {\(b_1\)};
    \node[above=8mm of b1] (bdots) {\(\vdots\)};
    \node[cartel,above=8mm of bdots] (bc) {\(b_c\)};
    \draw (s)--(v1)--(dots)--(vq);
    \draw[dashed] (s)--(b1)--(v1);
    \draw[dashed] (s)--(bc)--(v1);
  \end{tikzpicture}
  \caption{A schematic tight instance for \(2\leq c\leq n-3\).
  The solid vertices form the honest path \(s=v_0,v_1,\ldots,v_q=t\), where \(q=n-c-1\).
  Each cartel vertex is adjacent to both \(s\) and \(v_1\); omitted cartel vertices have the same adjacencies.
  The construction itself also covers the boundary cases.}
  \label{fig:tight-instance}
\end{figure}
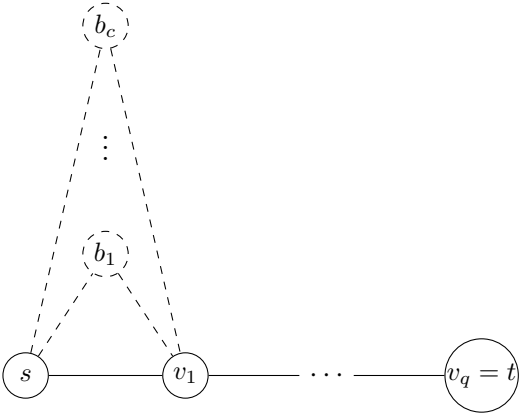

\subsection{Independent executions}

Suppose that \(C\) remains fixed, that \(k\) executions use mutually independent fresh rankings, and set \(M=\binom{n-1}{c}\).
By \cref{thm:lower-bound}, the material of each execution is intercepted with probability at most \(1-M^{-1}\), so independence gives
\[
  \Pr(\text{all \(k\) routes intercepted})
  \leq
  \left(1-\frac{1}{M}\right)^{\!k},
\]
with equality for the construction in \cref{thm:tightness}.

\begin{corollary}[Constant-success threshold]
\label{cor:reruns}
For \(1\leq c\leq n-2\) and fixed \(q\in(0,1)\), the worst-case number of independent executions needed to make the probability of at least one cartel-free route at least \(q\) is
\[
  k_q
  =
  \left\lceil
    \frac{\log(1-q)}{\log(1-1/M)}
  \right\rceil
  =
  \Theta(M),
\]
which for every fixed \(c\) equals \(\Theta\!\binom{n-1}{c}=\Theta(n^c)\).
\end{corollary}

\begin{proof}
The exact expression follows by solving \((1-1/M)^k\leq 1-q\), and it is necessary on the tight instance.
Since \(-\log(1-1/M)=\Theta(1/M)\), the expression is \(\Theta(M)\).
For fixed \(c\), the product formula \(\binom{n-1}{c}=(n-1)(n-2)\cdots(n-c)/c!\) gives \(\Theta(n^c)\).
\end{proof}

The necessity statement concerns independent nonadaptive reruns of \(\RPFR\); it is not a lower bound for arbitrary routing algorithms or for a protocol that learns cartel identities between executions.

\section{Comparisons and cost}
\label{sec:implications}

In the published \(\HS\) rule, each fragment carries a fresh score seed and visit history \citep[Secs.~3.2 and~3.8]{petrucena2026topology}.
The current relay forwards directly to \(t\) when \(t\) is adjacent; otherwise it selects the highest-scored unvisited neighbor and falls back to a non-backtracking step when every neighbor has been visited.
Its scored decisions therefore range over the current neighborhood rather than the entire frontier.

\begin{proposition}[Exponential local failure]
\label{prop:local-failure}
For every \(m\geq 2\) there is a biconnected graph on \(n=2m+2\) vertices whose single cartel vertex \(w\) is adjacent to every other vertex, such that in the ideal i.i.d.\ continuous-score model local \(\HS\) avoids \(w\) with probability
\[
  \frac{2^{m-1}\bigl((m-1)!\bigr)^2}{(2m-1)!}
  \sim
  \frac{\sqrt{\pi}}{2^{m}\sqrt{m}}
  =2^{-\Omega(n)}.
\]
\end{proposition}

\begin{proof}
Take a path \(s=x_0,x_1,\ldots,x_m=t\), attach a leaf \(y_i\) to \(x_i\) for each \(i\in\{0,\ldots,m-1\}\), and join \(w\) to every other vertex, as in \cref{fig:local-counterexample}.
Deleting \(w\) leaves a tree, and deleting any other vertex leaves \(w\) adjacent to all survivors, so \(G\) is biconnected and \(\{w\}\) is admissible.

Because only ranks matter, the probability-integral transform lets us assume that the scores are independent and uniform on \([0,1]\).
Condition on \(R(w)=z\).
Whenever the walk stands at \(x_i\) with \(i\leq m-2\), the target is not adjacent and the unvisited neighbors are exactly \(x_{i+1}\), \(y_i\), and \(w\).
Choosing \(w\) is immediate interception, and choosing \(y_i\) forces it, since \(w\) is then the only unvisited neighbor of the leaf.
The walk therefore survives step \(i\) precisely when \(R(x_{i+1})>\max\{R(y_i),z\}\), of conditional probability \((1-z^2)/2\).
At \(x_{m-1}\) the target is adjacent and the direct-target rule takes the final edge, so that step costs no comparison.
The \(m-1\) surviving events involve pairwise disjoint score pairs \(\{x_{i+1},y_i\}\) and are conditionally independent given \(z\), so
\[
  \Pr(\text{avoid }w)
  =
  \int_0^1\!\bigl(\tfrac{1-z^2}{2}\bigr)^{m-1}dz
  =
  \frac{2^{m-1}\bigl((m-1)!\bigr)^2}{(2m-1)!}.
\]
Stirling's formula gives the asymptotic, and \(m=(n-2)/2\) makes it \(2^{-\Omega(n)}\).
\end{proof}

\begin{figure}[t]
  \centering
  \resizebox{\columnwidth}{!}{%
  \begin{tikzpicture}[
    vertex/.style={circle,draw,fill=white,minimum size=7mm,inner sep=1pt},
    cartel/.style={circle,draw,dashed,fill=white,minimum size=7mm,inner sep=1pt},
    uedge/.style={densely dotted,gray}
  ]
    \node[vertex] (x0) {\(s\)};
    \node[vertex,right=12mm of x0] (x1) {\(x_1\)};
    \node[right=10mm of x1] (xd) {\(\cdots\)};
    \node[vertex,right=10mm of xd] (xm1) {\(x_{m-1}\)};
    \node[vertex,right=12mm of xm1] (t) {\(t\)};
    \node[vertex,above=9mm of x0] (y0) {\(y_0\)};
    \node[vertex,above=9mm of x1] (y1) {\(y_1\)};
    \node[above=9mm of xd] (yd) {\(\cdots\)};
    \node[vertex,above=9mm of xm1] (ym1) {\(y_{m-1}\)};
    \node[cartel,below=15mm of xd] (w) {\(w\)};
    \foreach \v in {x0,x1,xm1,t,y0,y1,ym1} \draw[uedge] (w)--(\v);
    \draw (x0)--(x1)--(xd)--(xm1)--(t);
    \draw (x0)--(y0);
    \draw (x1)--(y1);
    \draw (xm1)--(ym1);
  \end{tikzpicture}
  }
  \caption{The counterexample of \cref{prop:local-failure}.
  The cartel vertex \(w\) is adjacent to every other vertex.
  At each \(x_i\) with \(i\leq m-2\), taking the leaf \(y_i\) forces the next step into \(w\).}
  \label{fig:local-counterexample}
\end{figure}
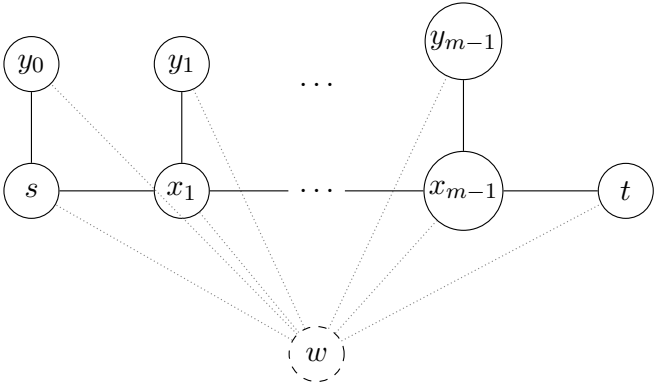

Because \(\{w\}\) is admissible, \cref{thm:lower-bound} gives \(\RPFR\) an avoidance probability of at least \((n-1)^{-1}\) on this very graph, against \(2^{-\Omega(n)}\) for local \(\HS\).
The proof of \cref{thm:lower-bound} fails locally because a high-scored honest vertex elsewhere in \(F(S)\) cannot protect the current vertex from a low-scored local trap.

\subsection{A topology-aware baseline}

The guarantee is weak against a rule that already knows \(G\).
For nonadjacent \(s,t\), let \(\kappa=\kappa(s,t)\) be the maximum number of internally vertex-disjoint \(s\)--\(t\) paths.
A controller holding the topology can select one of these paths uniformly at random, and since a \(c\)-cartel meets at most \(c\) of them, it avoids the cartel with probability at least \(1-c/\kappa\).
When \(\kappa>c\), this can be far larger than \(\binom{n-1}{c}^{-1}\).

\(\RPFR\) is not competing in that setting: it discovers adjacency only from visited vertices, so it cannot enumerate disjoint paths before routing, and \cref{thm:lower-bound} needs no connectivity input at all.
Where a trusted and current map of the network is available, disjoint-path selection dominates; the value of a topology-oblivious rule lies in the cases where it is not.

\subsection{State and cost}

\(\RPFR\) requires persistent execution state but not stateful relays.
An authenticated token can carry the visited set, discovered topology and frontier, random ranking, and enough predecessor information to move through \(G[S]\).
Relays can therefore be stateless, at the cost of a token that may grow to \(O(n+\lvert E\rvert)\) size.
With adjacency lists, a frontier-membership set, and a max-priority queue, one execution takes \(O(\lvert E\rvert+n\log n)\) time, stores \(O(n+\lvert E\rvert)\) topology, expands at most \(n-1\) vertices, and traverses \(O(n^2)\) edges when relocation follows a maintained spanning tree of \(G[S]\).
These are implementation upper bounds, and at the worst-case threshold \(k=\Theta\!\binom{n-1}{c}\) both inherit the binomial factor.

\bibliographystyle{plainnat}
\bibliography{refs}

\typeout{get arXiv to do 4 passes: Label(s) may have changed. Rerun}
\end{document}